\documentclass[twocolumn,10pt]{article}
\usepackage[T1]{fontenc}
\usepackage[utf8]{inputenc}
\usepackage{mathptmx} 
\usepackage[scaled=.90]{helvet} 
\usepackage{courier} 

\usepackage[a4paper,top=0.6in,bottom=0.6in,left=0.6in,right=0.6in]{geometry}
\usepackage{graphicx}
\usepackage{xcolor}
\usepackage{cuted}
\usepackage{microtype}
\usepackage{parskip}
\usepackage{tcolorbox}
\tcbuselibrary{skins,breakable}
\usepackage{academicons}
\usepackage{hyperref}
\usepackage{lipsum}
\usepackage{lastpage}
\usepackage{booktabs}
\usepackage{caption}
\usepackage{enumitem}
\usepackage{fancyhdr}
\usepackage{titlesec}
\usepackage{float}
\usepackage[ruled,vlined]{algorithm2e}
\usepackage{wrapfig}
\usepackage{amsmath, amssymb}
\usepackage{multicol}
\usepackage{amsthm}

\newtheorem{assumption}{Assumption}
\newtheorem{definition}{Definition}
\newtheorem{theorem}{Theorem}

\renewcommand{\qedsymbol}{}  

\usepackage[backend=biber,style=ieee]{biblatex}
\definecolor{primary}{HTML}{003049}      
\definecolor{accent}{HTML}{F77F00}       
\definecolor{textgray}{HTML}{2F2F2F}
\definecolor{dividergray}{HTML}{DADADA}

\hypersetup{
  colorlinks=true,
  linkcolor=primary,
  citecolor=primary,
  urlcolor=primary
}

\titleformat{\section}
  {\color{primary}\sffamily\Large\bfseries}
  {\thesection}{1em}{}[\vspace{0.2em}\titlerule]

\titleformat{\subsection}
  {\color{primary}\sffamily\large\bfseries}
  {\thesubsection}{1em}{}

\titlespacing*{\section}{0pt}{6pt}{3pt}
\titlespacing*{\subsection}{0pt}{4pt}{2pt}

\setlist[itemize]{left=1.2em, itemsep=2pt, topsep=2pt, parsep=0pt}
\setlist[enumerate]{left=1.2em, itemsep=2pt, topsep=2pt, parsep=0pt}

\tcbset{
  abstractstyle/.style={
    enhanced,
    colback=white,
    colframe=primary,
    boxrule=1pt,
    fonttitle=\bfseries\sffamily\large,
    left=6mm, right=6mm, top=2mm, bottom=2mm,
    width=\textwidth,
    boxsep=4pt,
    breakable
  },
  biobox/.style={
    colback=white,
    colframe=primary,
    arc=1.5mm,
    boxrule=0.5pt,
    drop shadow southeast,
    left=3mm, right=3mm, top=1mm, bottom=1mm,
    width=\linewidth,
    before skip=6pt, after skip=6pt,
    breakable
  }
}

\SetAlgoNlRelativeSize{-1}
\SetAlCapNameFnt{\sffamily\bfseries\small}
\SetAlCapFnt{\sffamily\small}

\makeatletter
\def\@maketitle{%
  \begin{center}
    {\fontsize{26pt}{20pt}\selectfont \bfseries \textcolor{primary}{On the Intractability of Chaotic Symbolic Walks: Toward a Non-Algebraic Post-Quantum Hardness Assumption}}\\[1ex]
    {\normalsize 
\textbf{
Mohamed Aly Bouke\,%
\href{https://orcid.org/0000-0003-3264-601X}{\includegraphics[height=1.8ex]{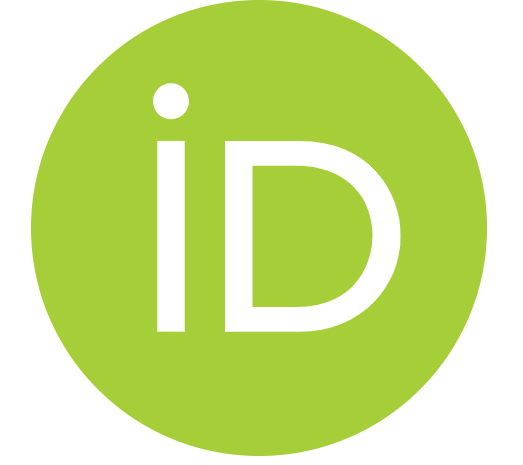}}%
\textsuperscript{1,*}
}
    }\\[0.8ex]
    {\footnotesize
      \textsuperscript{1}Department of Communication Technology and Network,\par Faculty of Computer Science and Information Technology, \par Universiti Putra Malaysia, Serdang 43400, Malaysia
    }\\[0.5ex]
    {\scriptsize \texttt{*bouke@ieee.org}}\\[1ex]
    {\scriptsize \textit{Short communication paper} — \today}
  \end{center}
}
\renewcommand{\maketitle}{%
  \twocolumn[
    \color{textgray}
    \@maketitle
    \vspace{-1.2em}
  ]
}
\makeatother

\begin{document}
\maketitle

\begin{strip}
  \begin{center}
    \begin{tcolorbox}[abstractstyle, title=Abstract]
      
      \normalsize
        Most classical and post-quantum cryptographic assumptions, including integer factorization, discrete logarithms, and Learning with Errors (LWE), rely on algebraic structures such as rings or vector spaces. While mathematically powerful, these structures can be exploited by quantum algorithms or advanced algebraic attacks, raising a pressing need for structure-free alternatives. To address this gap, we introduce the \emph{Symbolic Path Inversion Problem (SPIP)}, a new computational hardness assumption based on symbolic trajectories generated by contractive affine maps with bounded noise over $\mathbb{Z}^2$. Unlike traditional systems, SPIP is inherently non-algebraic and relies on chaotic symbolic evolution and rounding-induced non-injectivity to render inversion computationally infeasible. We prove that SPIP is PSPACE-hard and \#P-hard, and demonstrate through empirical simulation that even short symbolic sequences (e.g., $n = 3$, $m = 2$) can produce over 500 valid trajectories for a single endpoint, with exponential growth reaching $2^{256}$ paths for moderate parameters. A quantum security analysis further shows that Grover-style search offers no practical advantage due to oracle ambiguity and verification instability. These results position SPIP as a viable foundation for post-quantum cryptography that avoids the vulnerabilities of algebraic symmetry while offering scalability, unpredictability, and resistance to both classical and quantum inversion.

      \vspace{0.5em}
      \textbf{Keywords:} Post-quantum cryptography, Symbolic dynamics, One-way functions, Structure-free cryptography, Chaotic maps.
Fractal trajectories
    \end{tcolorbox}
  \end{center}
\end{strip}

\section{Introduction} \label{sec:intro}
\vspace{0.8em}

Modern public-key cryptography relies heavily on algebraic hardness assumptions, including problems such as integer factorization, discrete logarithms, and lattice-based constructions like Learning with Errors (LWE). These assumptions enable powerful reductions and efficient implementations, but they also possess mathematical regularity, rings, fields, and structured polynomials that may ultimately be exploited by advanced algebraic or quantum algorithms.

The emergence of quantum computing has raised concerns about the long-term viability of such structured assumptions. Shor’s algorithm, in particular, renders RSA and ECC insecure by solving their core problems in polynomial time \cite{shor1994algorithms}. While lattice-based systems currently offer promising post-quantum resilience, they remain embedded within algebraic frameworks that could, in principle, be compromised by future quantum advances or novel cryptanalytic techniques \cite{hecht2021pqc,alagic2024status,aharonov2005lattice}.

These challenges motivate a fundamental question: Can we construct cryptographic hardness assumptions that do not depend on algebraic structures at all? This paper proposes an affirmative answer, introducing a new class of hardness assumptions grounded in the dynamics of symbolic chaos. Specifically, we investigate whether the inversion of symbolic paths, generated by contractive affine transformations with bounded noise and discrete rounding, can serve as a foundation for a fundamentally different form of computational hardness, distinct from classical algebraic paradigms.

This work does not emerge in isolation, but builds upon a prior mathematical investigation of chaotic systems in the form of Random Nonlinear Iterated Function Systems (RNIFS) \cite{bouke2025rnifs}. That foundational study rigorously analyzes the properties of symbolic trajectories, bounded noise, and fractal attractors in discrete chaotic maps. It demonstrates how noise-induced perturbations and rounding operations create combinatorial ambiguity and symbolic unpredictability, leading to complex, non-invertible systems. In this paper, we leverage those mathematical insights to define the \textit{Symbolic Path Inversion Problem} (SPIP) as a cryptographic hardness assumption. Together, these works form a coherent research thread aimed at exploring structure-free, chaos-based approaches to post-quantum cryptography.

Our objective here is not to design a full cryptographic scheme, but rather to define and justify a standalone computational assumption: that inverting symbolic chaotic trajectories is infeasible for any efficient algorithm. We formally define the symbolic transformation model, analyze its combinatorial and dynamical properties, and establish hardness results via reductions from well-known complexity theoretic problems. This theoretical foundation is intended to support future research into cryptographic primitives rooted in symbolic dynamics, offering a fundamentally different hardness landscape for the quantum era.

\subsection{Positioning SPIP}

The SPIP proposed in this paper is intended as a novel computational hardness assumption, conceptually parallel to foundational problems such as integer factorization (used in RSA) \cite{rivest1984rsa}, the discrete logarithm problem (used in ECC) \cite{merkle1978secure}, the LWE problem \cite{regev2009lattices,lyubashevsky2010ideal}, and multivariate polynomial systems (e.g., HFE, Rainbow) \cite{li1997numerical,ding2005rainbow}.

Whereas these traditional assumptions are rooted in algebraic structures such as finite fields, polynomial rings, or vector spaces, SPIP is built on symbolic dynamics over the integer lattice~$\mathbb{Z}^2$. This difference becomes particularly significant in the post-quantum setting, where many known quantum algorithms, such as Shor’s, depend on algebraic regularity to achieve speed-ups.

SPIP instead aligns with a growing class of cryptographic approaches that draw hardness from combinatorial explosion, chaotic evolution, and symbolic unpredictability. To the best of our knowledge, this is the first formalization of SPIP as a cryptographic assumption, marking a new direction in post-quantum design based on symbolic complexity rather than algebraic structure.

\begin{tcolorbox}[title=\textbf{Conceptual Primer: Symbolic Dynamics and SPIP}, colframe=accent, coltitle=black, fonttitle=\bfseries]
\textbf{Symbolic Dynamics:} A mathematical framework where the state of a system evolves through discrete symbols (e.g., integers or letters), rather than real-valued equations. In SPIP, each symbol represents an affine transformation applied to a point in the grid.

\textbf{Affine Transformation:} A function of the form $T(x) = A x + b$, where $A$ is a matrix and $b$ is a vector. It maps points to new positions by rotating, scaling, and shifting them.

\textbf{Contractive Map:} A transformation that brings points closer together. This ensures trajectories don’t diverge uncontrollably, but rather cluster.

\textbf{Rounding (Floor Function):} After applying a transformation, we round the output to the nearest integer point in $\mathbb{Z}^2$. This introduces symbolic ambiguity.

\textbf{Noise ($\delta$):} A small random vector added to each transformation step, making the system non-deterministic and amplifying path diversity.

\textbf{Symbolic Path:} A sequence of points $\{x_0, x_1, ..., x_n\}$ generated by applying transformations according to a symbolic code (e.g., $\sigma = (1,2,1)$).

\textbf{SPIP:} The challenge of recovering any valid symbolic path from the start and end points is the basis for our proposed cryptographic hardness assumption.
\end{tcolorbox}

\section{Mathematical Preliminaries} \label{sec:preliminaries}
\vspace{0.8em}

We begin by formalizing the symbolic dynamical model that underlies our hardness assumption. The system is defined over the two-dimensional integer lattice $\mathbb{Z}^2$ and evolves through a discrete sequence of contractive affine transformations, perturbed by bounded random noise. Each transformation step produces a symbolic trajectory point, forming a path whose inversion constitutes the core cryptographic challenge.

To more precisely characterize the evolution of symbolic trajectories, we refine the model by constraining the transformation sequence $\mathcal{T} = \{T_i\}_{i=0}^{n-1}$ to be selected from a finite collection of contractive affine maps, denoted $\mathbb{T} = \{T^{(1)}, T^{(2)}, ..., T^{(m)}\}$. Each $T^{(j)}(x) = \lfloor A^{(j)} x + b^{(j)} + \delta \rfloor$ is defined by a fixed contractive matrix $A^{(j)}$ and translation vector $b^{(j)}$, where $\|A^{(j)}\|_2 < 1$ and $b^{(j)} \in \mathbb{R}^2$.

The specific transformation $T_i$ applied at step $i$ is determined by a symbolic code $\sigma_i \in \Sigma = \{1, 2, ..., m\}$, so that:
\[
T_i := T^{(\sigma_i)}
\]
This symbolic sequence $\sigma = (\sigma_0, ..., \sigma_{n-1}) \in \Sigma^n$ defines the transformation path, analogous to a control code in Iterated Function Systems (IFS) or symbolic subshifts. The evolution of the trajectory thus follows:
\[
x_{i+1} = T^{(\sigma_i)}(x_i) = \lfloor A^{(\sigma_i)} x_i + b^{(\sigma_i)} + \delta_i \rfloor
\]

This refinement enables structured modeling of the transformation process and aligns SPIP with well-established paradigms in symbolic dynamics. It also ensures sufficient entropy and branching complexity by allowing $m > 1$ symbolic options per step, increasing the combinatorial diversity of feasible paths.

\subsection{Contractive Affine Maps over $\mathbb{Z}^2$}

Let $\mathcal{T} = \{T_i\}_{i=0}^{n-1}$ be a finite sequence of affine transformations acting on $\mathbb{Z}^2$, where each map $T_i$ is defined as:
\[
T_i(x) = \lfloor A_i x + b_i + \delta_i \rfloor
\]
With:
\begin{itemize}
    \item $A_i \in \mathbb{R}^{2 \times 2}$ is a contractive matrix: $\|A_i\|_2 < 1$.
    \item $b_i \in \mathbb{R}^2$ is a deterministic translation vector.
    \item $\delta_i \in [-\epsilon, \epsilon]^2$ is a bounded noise vector, sampled independently for each $i$.
    \item $\lfloor \cdot \rfloor$ denotes component-wise floor rounding to the nearest integer coordinate in $\mathbb{Z}^2$.
\end{itemize}

Contractiveness ensures that small variations in the input amplify over time into symbolic divergence, while rounding and noise introduce nonlinearity and symbolic ambiguity.

Given the finite transformation set $\mathbb{T} = \{T^{(1)}, T^{(2)}, ..., T^{(m)}\}$ as introduced above, we define a symbolic alphabet $\Sigma = \{1, 2, ..., m\}$ representing the indices of available transformations. Each symbolic trajectory can therefore be described not only in terms of its state-space path $\mathcal{P} = \{x_0, x_1, ..., x_n\}$, but also via its associated symbolic code:
\[
\sigma = (\sigma_0, \sigma_1, ..., \sigma_{n-1}) \in \Sigma^n
\]
where $T_i := T^{(\sigma_i)}$ at step $i$. This symbolic code serves as a discrete encoding of the transformation sequence that governs the trajectory evolution.

Furthermore, while the SPIP model relies on contractive affine maps of the form $T(x) = \lfloor A x + b + \delta \rfloor$, this local use of affine structure does not contradict our broader claim of being "structure-free" in the cryptographic sense. In traditional cryptography, hardness assumptions often derive from well-behaved algebraic objects such as groups, rings, or fields, structures that possess closure, associativity, and invertibility, and are thus amenable to symbolic manipulations or quantum algorithms exploiting hidden symmetries (e.g., Shor's or index calculus methods).

In contrast, SPIP's use of affine transformations is severely disrupted by two critical mechanisms: (1) bounded random noise $\delta \in [-\epsilon, \epsilon]^2$, and (2) nonlinear discretization via the floor function. These mechanisms introduce non-differentiability, state-space discontinuities, and path non-uniqueness, which together obliterate the algebraic tractability of the underlying affine maps.

More precisely, even though $A$ and $b$ are linearly defined, the overall transformation $x \mapsto \lfloor A x + b + \delta \rfloor$ is not injective, not continuous, and not reversible, and its composition over multiple steps amplifies symbolic ambiguity exponentially. Therefore, the resulting symbolic trajectory space does not admit algebraic reductions or structured analysis, rendering SPIP fundamentally non-algebraic in its global behavior, even if affine primitives are used locally.

This distinction aligns with the goal of constructing post-quantum assumptions that resist algebraic exploitation, the presence of localized linear components does not imply that the system as a whole exhibits algebraic structure exploitable by adversaries.

\subsection{Symbolic Trajectories}
Given an initial point $x_0 \in \mathbb{Z}^2$, the symbolic trajectory $\mathcal{P}$ of length $n$ is defined recursively by:
\[
x_{i+1} = T_i(x_i) = \lfloor A_i x_i + b_i + \delta_i \rfloor
\]
for $i = 0, 1, \dots, n-1$. The resulting path $\mathcal{P} = \{x_0, x_1, ..., x_n\}$ evolves deterministically concerning the transformations and the sampled noise.

Each $x_i$ can be seen as a symbol in a trajectory over the space $\mathbb{Z}^2$, and the entire path represents a discrete symbolic encoding of a chaotic process. The symbolic growth is highly sensitive to the initial condition $x_0$ and to perturbations introduced at each step.

\subsection{Noise Distribution and Entropy Source}
The noise vectors $\delta_i$ are assumed to be drawn from a bounded, uniform distribution:
\[
\delta_i \sim \mathcal{U}([-\epsilon, \epsilon]^2)
\]
for some small fixed $\epsilon > 0$.

The noise serves two purposes:
\begin{itemize}
    \item It breaks determinism and injects entropy into the symbolic path, increasing the difficulty of inversion.
    \item It ensures that multiple symbolic paths may lead to the same endpoint $x_n$, amplifying combinatorial ambiguity.
\end{itemize}

The combination of contractiveness, rounding, and bounded noise results in a trajectory space with exponential growth in the number of valid symbolic paths a key element in establishing the hardness of the inversion problem in subsequent sections.

\section{Symbolic Path Inversion Problem} \label{sec:spip}
\vspace{0.8em}

We now formalize the core inversion problem associated with symbolic chaotic trajectories. This inversion task constitutes the basis for our proposed hardness assumption and serves as the central computational challenge in this paper.

\subsection{Problem Definition}

Let $\mathcal{T} = \{T_i\}_{i=0}^{n-1}$ be a fixed sequence of contractive affine transformations as described in Section~\ref{sec:preliminaries}, and let $x_0 \in \mathbb{Z}^2$ be the initial point of a symbolic trajectory of length $n$. The symbolic path evolves according to:
\[
x_{i+1} = T_i(x_i) = \lfloor A_i x_i + b_i + \delta_i \rfloor
\quad \text{for } i = 0, 1, \dots, n-1
\]
Let $x_n$ be the endpoint of the trajectory. The inversion problem is defined as follows:

\begin{definition}[SPIP]
Given the initial and final points $x_0, x_n \in \mathbb{Z}^2$, the transformation sequence $\mathcal{T}$, and the trajectory length $n$, find any valid symbolic path $\mathcal{P} = \{x_0, x_1, \dots, x_n\}$ consistent with the transformation rules and noise bounds.
\end{definition}

This formulation abstracts away any cryptographic hashing or encoding layers. Our goal is to evaluate the inherent difficulty of reconstructing a valid symbolic trajectory from its endpoints and transformation sequence alone.

\subsection{Problem Characteristics}

The SPIP exhibits several key properties that contribute to its inversion hardness:

Due to the symbolic rounding and bounded noise introduced at each step, the number of valid intermediate states $x_i$ grows combinatorially with $n$. Even under fixed transformations $\mathcal{T}$, the symbolic ambiguity induced by $\delta_i$ and discretization leads to an exponentially large set of feasible trajectories. This yields a path space of size $\Omega(m^n)$ for some constant $m > 1$, which precludes exhaustive enumeration for moderate $n$.

Distinct trajectories may converge to the same endpoint $x_n$, a phenomenon we refer to as \emph{symbolic overlap}. This many-to-one mapping arises from the rounding and perturbation noise, rendering the inversion problem non-injective. In particular, there may exist exponentially many valid trajectories mapping $(x_0, \mathcal{T}) \rightarrow x_n$.

Importantly, the difficulty of SPIP does not depend on cryptographic hashing. The hardness is intrinsic to the symbolic dynamics: recovering even a single valid trajectory from $x_0$, $x_n$, and $\mathcal{T}$ requires navigating a non-deterministic transformation space whose symbolic ambiguity and path collisions make deterministic reconstruction computationally infeasible.

These characteristics, taken together, motivate our forthcoming formalization of the SPIP Hardness Assumption as a non-algebraic foundation for post-quantum cryptography.

\subsection{Illustrative Numerical Example}

To concretely demonstrate the practical difficulty of inverting symbolic trajectories in SPIP, we present a simplified numerical example with minimal parameters. This example highlights how even short symbolic paths with limited transformations and small noise can produce significant ambiguity, rendering inversion computationally infeasible.

\paragraph{Setup:}
We consider the following configuration:
\begin{itemize}
  \item Initial point: $x_0 = (0, 0)$
  \item Number of steps: $n = 3$
  \item Transformation count: $m = 2$
  \item Noise bound: $\epsilon = 0.5$
  \item Contractive matrix: $A = 0.5 \cdot I_2$
  \item Transformations:
  \begin{align*}
    T^{(1)}(x) &= \left\lfloor A x + (1, 0) + \delta \right\rfloor \\
    T^{(2)}(x) &= \left\lfloor A x + (0, 1) + \delta \right\rfloor
  \end{align*}
\end{itemize}

The total number of symbolic sequences of length $n=3$ from $m=2$ transformations is:
\[
|\Sigma^3| = m^n = 2^3 = 8
\]
Each symbolic sequence $\sigma = (\sigma_0, \sigma_1, \sigma_2)$ represents a different control path over the transformation set. For each $\sigma$, the bounded noise $\delta_i \in [-0.5, 0.5]^2$ introduces variability in the trajectory outcome due to rounding.

Let us estimate the number of distinct outcomes from a single transformation. Assume the perturbed region $Ax + b + \delta$ spans a square of width $1$ in $\mathbb{R}^2$. The number of integer lattice points covered by this region, after applying $\lfloor \cdot \rfloor$, can be up to $4$ distinct integer outputs per transformation (as floor can land on 2 values per coordinate).

Thus, the effective branching factor per step is at least:
\[
k \geq 4
\]
This leads to a lower bound on the total number of distinct possible trajectories of:
\[
|\mathcal{W}_3| \geq (m \cdot k)^n = (2 \cdot 4)^3 = 8^3 = 512
\]
So, even with only 2 transformations and 3 steps, there can be \textbf{at least 512} valid symbolic trajectories, many of which may collide at the same endpoint $x_3$.

Assume a randomly chosen symbolic sequence $\sigma = (1, 2, 1)$ and sampled noise vectors:
\[
\delta_1 = (0.3, -0.4), \quad \delta_2 = (-0.2, 0.2), \quad \delta_3 = (0.1, 0.4)
\]

We compute the trajectory:
\begin{align*}
x_1 &= \left\lfloor 0.5 \cdot (0, 0) + (1, 0) + (0.3, -0.4) \right\rfloor = (1, -1) \\
x_2 &= \left\lfloor 0.5 \cdot (1, -1) + (0, 1) + (-0.2, 0.2) \right\rfloor = \left\lfloor (0.3, 0.7) \right\rfloor = (0, 0) \\
x_3 &= \left\lfloor 0.5 \cdot (0, 0) + (1, 0) + (0.1, 0.4) \right\rfloor = \left\lfloor (1.1, 0.4) \right\rfloor = (1, 0)
\end{align*}

Final output: $x_3 = (1, 0)$

An adversary trying to recover the intermediate trajectory from $(x_0, x_3)$ and the known transformation set faces:
\begin{itemize}
  \item $|\Sigma^3| = 8$ symbolic sequences
  \item For each, at least $4^3 = 64$ noise-induced discrete paths
  \item Total symbolic ambiguity: $\geq 8 \times 64 = 512$ candidate paths
\end{itemize}

Many of these 512 paths may lead to the same endpoint $x_3$, due to rounding and contraction. However, from the final point alone, the adversary cannot distinguish which symbolic path or noise realization was taken. Thus, inversion becomes computationally infeasible, even for this minimal configuration.

\section{Complexity Theoretic Analysis} \label{sec:complexity}
\vspace{0.8em}

To justify the computational hardness of the SPIP, we now analyze its complexity from a theoretical standpoint. We establish its relationship to well-known problems in symbolic dynamics and formal verification, demonstrating that SPIP is at least as hard as classic problems in the \#P and PSPACE complexity classes.

\subsection{PSPACE-Hardness}

We now formally prove that SPIP is PSPACE-hard by reduction from the Symbolic Reachability problem, which is known to be PSPACE-complete \cite{arora2009computational}.

\begin{definition}[Symbolic Reachability Problem]
Given a finite symbolic transition system $\Sigma = (S, T, s_0, s_t)$, where:
\begin{itemize}
  \item $S$ is a finite symbolic state space,
  \item $T \subseteq S \times S$ is a set of transitions,
  \item $s_0, s_t \in S$ are the initial and target states,
\end{itemize}
decide whether there exists a sequence of transitions $s_0 \rightarrow s_1 \rightarrow \cdots \rightarrow s_t$ of length $\leq n$ such that each $(s_i, s_{i+1}) \in T$.
\end{definition}

\begin{theorem}
The SPIP is PSPACE-hard.
\end{theorem}

\begin{proof}
We reduce from the Symbolic Reachability problem. Let $\Sigma = (S, T, s_0, s_t)$ be a symbolic system as defined. We construct an SPIP instance $\mathcal{I}$ such that solving $\mathcal{I}$ allows us to solve $\Sigma$.

\textbf{Step 1: Encoding States.} For each symbolic state $s_i \in S$, define a unique point $x_i \in \mathbb{Z}^2$. This encoding can be achieved by assigning integer lattice positions using a bijective mapping $\phi: S \rightarrow \mathbb{Z}^2$.

\textbf{Step 2: Encoding Transitions.} For each symbolic transition $(s_i, s_j) \in T$, define a contractive affine transformation $T_{(i,j)}(x) = \lfloor A_{(i,j)} x + b_{(i,j)} + \delta_{(i,j)} \rfloor$, where:
\begin{itemize}
  \item $A_{(i,j)}$ is a fixed contractive matrix (e.g., $0.5I$),
  \item $b_{(i,j)}$ is chosen such that $T_{(i,j)}(\phi(s_i)) \approx \phi(s_j)$,
  \item $\delta_{(i,j)}$ is bounded noise allowing minor deviation, ensuring convergence to $x_j$ under floor rounding.
\end{itemize}

\textbf{Step 3: Constructing Trajectory.} Given a path from $s_0$ to $s_t$ in $\Sigma$, the sequence of symbolic states $(s_0, s_1, \dots, s_t)$ corresponds to a valid symbolic trajectory $(x_0, x_1, \dots, x_t)$ in SPIP.

\textbf{Step 4: Inversion Equivalence.} Solving the SPIP instance, i.e., recovering a valid path $\{x_0, \dots, x_t\}$ consistent with the affine maps and noise, allows recovery of the original symbolic path in $\Sigma$. Therefore, solving SPIP solves Symbolic Reachability.

\textbf{Step 5: Polynomial Reduction.} The encoding steps (bijection, matrix selection, affine setup) are all computable in polynomial time concerning the size of $\Sigma$.

Thus, SPIP is PSPACE-hard under polynomial-time reductions.
\end{proof}

\subsection{\#P-Hardness}

We now formally establish that the counting version of SPIP is \#P-hard by reducing from the classical \#PATH COUNT problem \cite{valiant1979complexity, sipser1996introduction}.

\begin{definition}[\#PATH COUNT Problem]
Given a directed acyclic graph (DAG) $G = (V, E)$ with distinguished nodes $s, t \in V$, compute the number of distinct paths from $s$ to $t$.
\end{definition}

\begin{theorem}
The problem of counting valid SPIP trajectories (SPIP-Count) is \#P-hard.
\end{theorem}

\begin{proof}
Let $G = (V, E)$ be an instance of \#PATH COUNT. Our goal is to construct an SPIP instance in which the number of valid symbolic trajectories from $x_0$ to $x_n$ corresponds exactly to the number of $s \rightarrow t$ paths in $G$.

\textbf{Step 1: State Encoding.} Assign each vertex $v_i \in V$ a unique symbolic state $x_i \in \mathbb{Z}^2$ via a bijection $\phi: V \rightarrow \mathbb{Z}^2$.

\textbf{Step 2: Edge Encoding.} For each edge $(u,v) \in E$, define a contractive affine map $T_{(u,v)}$ such that:
\[
T_{(u,v)}(x_u) = \lfloor A x_u + b_{(u,v)} + \delta_{(u,v)} \rfloor \approx x_v
\]
Where:
\begin{itemize}
    \item $A$ is a fixed contractive matrix (e.g., $0.5I$),
    \item $b_{(u,v)}$ is chosen so that the unperturbed map lands near $x_v$,
    \item $\delta_{(u,v)}$ is bounded noise (e.g., in $[-1,1]^2$) allowing rounding to $x_v$.
\end{itemize}

\textbf{Step 3: Trajectory Construction.} A valid symbolic trajectory in SPIP corresponds to a valid path in $G$. Each sequence of transformations that maps $x_0 = \phi(s)$ to $x_n = \phi(t)$ represents a symbolic realization of a path from $s$ to $t$.

\textbf{Step 4: One-to-One Mapping.} For each path $\pi = (s \rightarrow v_1 \rightarrow \dots \rightarrow t)$ in $G$, there exists a unique (up to bounded noise) symbolic trajectory $\mathcal{P}_\pi$ in SPIP that applies the corresponding affine maps.

\textbf{Step 5: Reduction.} Since each valid path in $G$ maps to a distinct trajectory in SPIP, counting SPIP trajectories from $x_0$ to $x_n$ is equivalent to counting $s \rightarrow t$ paths in $G$. The reduction can be computed in polynomial time.

Hence, SPIP-Count is \#P-hard. \renewcommand{\qedsymbol}{}
\end{proof}

\subsection{Growth of the Solution Space}

These complexity results imply that any polynomial-time inversion algorithm for SPIP would imply breakthroughs in symbolic reachability and path counting, both long-standing intractable problems.

In practice, the solution space of SPIP exhibits a rapid combinatorial explosion. For example, assume that for each transformation $T_i$, the bounded noise $\delta_i$ can cause $k$ distinct rounded outcomes. Then, the number of possible symbolic trajectories of length $n$ is at least $k^n$. For modest values of $k$ (e.g., $k = 4$) and $n = 128$, the number of possible paths exceeds $2^{256}$, already infeasible to search or store.
This exponential trajectory growth, combined with symbolic overlap and non-injectivity, reinforces our position that SPIP constitutes a plausible and practically hard computational assumption suitable for post-quantum cryptography.

To support the claim of exponential trajectory growth, we now provide a lower bound analysis based on per-step ambiguity induced by bounded noise and rounding.

At each step $i$, the applied transformation is:
\[
x_{i+1} = \lfloor A^{(\sigma_i)} x_i + b^{(\sigma_i)} + \delta_i \rfloor
\]
Due to bounded noise $\delta_i \in [-\epsilon, \epsilon]^2$ and floor rounding, the perturbed input region maps to an uncertain discrete neighborhood in $\mathbb{Z}^2$. Let $k$ denote a conservative lower bound on the number of possible integer lattice outputs (distinct $x_{i+1}$ values) that can arise from a single $x_i$ under a fixed $T^{(j)}$ and varying $\delta_i$:
\[
k := \min_{x \in \mathbb{Z}^2, j \in \Sigma} \left|\left\{ \lfloor A^{(j)} x + b^{(j)} + \delta \rfloor \,\middle|\, \delta \in [-\epsilon, \epsilon]^2 \right\} \right|
\]

That is, $k$ quantifies the symbolic branching factor per step due to noise-induced discretization. Under this model, the number of distinct symbolic paths of length $n$ is lower bounded by:
\[
|\mathcal{W}_n| \geq k^n
\]
This bound holds even when the transformation sequence $\sigma = (\sigma_0, ..., \sigma_{n-1})$ is fixed. When $\sigma$ itself is variable and sampled from a symbolic alphabet $\Sigma$ of size $m$, the combined path space expands to:
\[
|\mathcal{W}_n| \geq (k \cdot m)^n
\]

Thus, symbolic path space exhibits exponential growth in $n$, stemming from both symbolic control (\textit{transformation selection}) and chaotic ambiguity (\textit{noise-induced branching}). This formally substantiates the hardness premise of SPIP, as no efficient algorithm can enumerate or invert exponentially many ambiguous paths.

\subsection{Hardness Assumption} \label{sec:spip-assumption}
\vspace{0.8em}

Building on the symbolic model and complexity analysis introduced in previous sections, we now formally state our main cryptographic assumption. This assumption posits that inverting symbolic chaotic trajectories, defined by a sequence of affine transformations with bounded noise, is computationally infeasible for any probabilistic polynomial, time adversary.

\begin{assumption}[Symbolic Path Inversion Hardness (SPIP-H)]
Let $\mathcal{T} = \{T_i\}_{i=0}^{n-1}$ be a sequence of contractive affine transformations on $\mathbb{Z}^2$ with bounded perturbation $\delta_i \in [-\epsilon, \epsilon]^2$, and let $x_0, x_n \in \mathbb{Z}^2$ be the start and end points of a symbolic trajectory $\mathcal{P} = \{x_0, x_1, \dots, x_n\}$, defined by:
\[
x_{i+1} = T_i(x_i) = \lfloor A_i x_i + b_i + \delta_i \rfloor
\quad \text{for } i = 0, 1, \dots, n-1
\]
We assume that no probabilistic polynomial-time (PPT) adversary $\mathcal{A}$, given $(x_0, x_n, \mathcal{T}, n)$, can recover any valid intermediate trajectory $\mathcal{P}'$ such that:
\[
\mathcal{P}' = \{x_0, x_1', ..., x_n\} \in \mathcal{W}_n(\mathcal{T})
\]
with non-negligible probability in $n$, where $\mathcal{W}_n(\mathcal{T})$ denotes the set of all symbolic paths under $\mathcal{T}$.
\end{assumption}

\subsection{Interpretation and Security Implications}

The SPIP-H assumption captures the intuition that symbolic chaotic paths, due to their exponential branching, rounding-induced non-injectivity, and noise perturbations, are inherently one-way in nature. This places SPIP-H in the family of hardness assumptions related to inverting structured combinatorial systems, albeit in a non-algebraic, structure-free setting.

Unlike traditional hardness assumptions based on algebraic structures (e.g., discrete logs, lattices, multivariate polynomials), SPIP-H does not rely on any homomorphic property or group-theoretic symmetry. This makes it naturally resistant to classes of quantum algorithms that exploit algebraic regularity, such as Shor’s and certain Grover-optimized search techniques.

SPIP-H serves as a plausible foundation for building one-way functions, pseudorandom generators, key derivation mechanisms, or even public-key schemes that are entirely structure-free. In later sections and follow-up work, we outline how additional layers, such as hash-based masking, can amplify symbolic unpredictability and embed SPIP into real cryptographic constructions.

\subsection{Quantum Security Considerations} \label{sec:quantum}
\vspace{0.8em}
A fundamental motivation behind the SPIP is to construct a post-quantum hardness assumption that resists quantum adversaries not by exploiting algebraic structure, but by avoiding it altogether. In the following paragraphs, we critically assess the security of SPIP against quantum algorithms, with a particular focus on Grover-style unstructured search and hybrid quantum-classical strategies.

Unlike factoring, discrete logarithms, and lattice-based schemes which reside in algebraic domains, SPIP does not embed its structure into groups, rings, or vector spaces. It operates purely over the symbolic composition of contractive affine maps with bounded noise and integer rounding. As such, there is no periodicity, hidden subgroup, or commutative algebra to exploit, rendering Shor’s algorithm and its variants inapplicable. This removes the most significant known exponential quantum threat.

Grover's algorithm can quadratically accelerate brute-force search in unstructured spaces \cite{grover1996fast}. In the context of SPIP, a quantum adversary could, in theory, evaluate candidate symbolic paths and test whether a given trajectory yields the observed endpoint $x_n$.

However, several structural barriers limit Grover’s applicability to SPIP:

\begin{itemize}
    \item \textbf{Exponential Search Space:} The total number of symbolic trajectories of length $n$ is $\Omega((m \cdot k)^n)$, where $m$ is the number of transformations and $k$ the minimum per-step symbolic branching due to noise. Thus, even with Grover's quadratic speedup, the time complexity is $\Omega(\sqrt{(m \cdot k)^n})$, which remains exponential in $n$.

    \item \textbf{Non-injective Evaluation:} SPIP's evaluation oracle is not efficiently reversible. Due to noise-induced overlap, many trajectories map to the same endpoint $x_n$, and verifying a candidate trajectory requires simulating stochastic transformations with bounded perturbation, a process inherently probabilistic and discretized.

    \item \textbf{Oracle Ambiguity:} Grover’s algorithm assumes a Boolean oracle indicating whether a candidate is correct. For SPIP, even this is ambiguous. Given an endpoint $x_n$, there may be hundreds of valid symbolic paths leading to it. Thus, the oracle’s success condition is fuzzy and non-unique, degrading Grover’s effectiveness.

    \item \textbf{Superposition Instability:} SPIP evaluation requires floor functions, discretization, and sampling from continuous noise. These operations are difficult to model coherently in a quantum superposition, and no efficient reversible circuit for SPIP evaluation currently exists. This obstructs the construction of a Grover-compatible quantum oracle.
\end{itemize}

Taken together, these characteristics suggest that SPIP behaves closer to a noisy combinatorial problem than an unstructured search task, complicating quantum evaluation models.

Furhtermore, SPIP contrasts with lattice-based schemes like LWE in two key ways: (1) LWE enjoys worst-case hardness proofs but is still reducible to linear algebraic structure, and (2) LWE security reductions often rely on hybrid arguments sensitive to decryption noise magnitude. SPIP, on the other hand, lacks such algebraic dependence and draws its hardness from symbolic path explosion, bounded noise, and topological ambiguity.

While this lack of structure prevents classic reductions, it also immunizes SPIP against known quantum decompositions. Indeed, as with certain isogeny-based protocols (e.g., CSIDH), SPIP’s complexity stems from the difficulty of path enumeration in highly non-linear, chaotic, and non-commutative systems.

\subsection{Parameter Considerations}

The hardness and cryptographic utility of SPIP are governed by three core tunable parameters: path length $n$, noise bound $\epsilon$, and the size of the symbolic transformation set $|\Sigma| = m$. Each plays a distinct and interdependent role in shaping the symbolic trajectory space $\mathcal{W}_n$, affecting both security guarantees and implementation costs.

The most direct contributor to hardness is the trajectory length $n$. Since the symbolic path space grows approximately as $(m \cdot k)^n$, where $k$ is the effective branching factor induced by bounded noise and rounding ambiguity, longer paths yield exponentially greater inversion difficulty. This growth not only increases the computational burden on an adversary but also introduces deeper symbolic collisions and more complex non-injectivity patterns, key ingredients for post-quantum resistance. However, increasing $n$ also introduces cost: forward trajectory computation, storage of intermediate states, and potential communication overhead in protocol contexts.

The noise parameter $\epsilon$ regulates symbolic branching. When $\epsilon$ is too small, symbolic ambiguity is limited, and trajectories become close to deterministic, making them easier to reverse. If $\epsilon$ becomes too large, the system risks degeneracy trajectory outputs saturate or lose resolution, diminishing their utility. There exists a critical range for $\epsilon$, often problem-dependent, where symbolic entropy is maximized without collapsing output variance. Empirically (see Section~\ref{sec:empirical}), this range lies roughly between 0.3 and 0.7, depending on the contraction of the affine maps.

The number of available affine maps, $m$, controls the symbolic alphabet size. A larger $m$ boosts per-step entropy and grows the path space multiplicatively. However, symbolic freedom (entropy per control bit) tends to decay when $m$ exceeds a threshold, due to overlapping outputs and map redundancy (see Figure~\ref{fig:symbolic_freedom}). Thus, while increasing $m$ helps, it must be done with attention to the geometric separation and contractive behavior of the maps.

The joint effect of $n$ and $\epsilon$ on symbolic space growth is visualized in Figure~\ref{fig:3d_entropy_plot}, where the approximate symbolic space size is shown in base-2 logarithmic scale:

\begin{figure}[htbp]
  \centering
  \includegraphics[width=\linewidth]{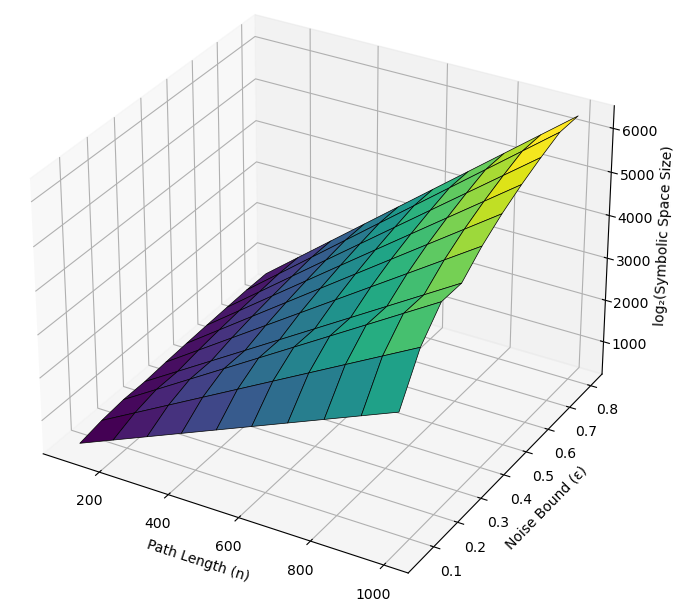}
  \caption{Effect of path length $n$ and noise bound $\epsilon$ on symbolic trajectory space size, assuming fixed $m = 10$. Growth is estimated as $\log_2((m \cdot k)^n)$, with $k \sim \lceil \epsilon \cdot 10 \rceil$.}
  \label{fig:3d_entropy_plot}
\end{figure}

From this surface plot, we observe the following:

\begin{itemize}
    \item Symbolic space size increases exponentially with $n$, as expected.
    \item Higher $\epsilon$ contributes multiplicatively via the branching factor $k$, but its effect saturates near $\epsilon = 0.7$, consistent with empirical entropy observations.
    \item The ridge near the upper-right corner reflects maximal trajectory ambiguity favorable for one-wayness but challenging for protocol determinism or symbolic encoding.
\end{itemize}

In conclusion, these parameters interact non-linearly, and their tuning determines both theoretical security and practical viability. Security grows exponentially in $n$, modulated by $m$ and $\epsilon$, but stability, efficiency, and symbolic resolution degrade if these parameters are unbalanced. Therefore, future cryptographic constructions based on SPIP must select parameter sets via careful profiling and simulation, possibly incorporating adaptive ranges or domain-specific calibration.

\section{Empirical Simulation} \label{sec:empirical}
\vspace{0.8em}

To complement the theoretical complexity analysis of SPIP presented in Sections~\ref{sec:spip} and~\ref{sec:complexity}, we now turn to an empirical investigation of symbolic trajectories under controlled parameter variations. The goal of this simulation study is not to benchmark specific cryptographic algorithms, but to probe the internal structure of the SPIP path space, including entropy growth, endpoint multiplicity, symbolic overlap, and collision patterns, and validate the assumptions underlying its one-wayness.

We simulate SPIP trajectories across eight experimental configurations, systematically increasing the number of transformation steps, the size of the symbolic control set, and the magnitude of the bounded noise. For each configuration, we record symbolic entropy, the number of unique endpoints, the most frequent endpoint count, symbolic freedom, and average spatial dispersion. These observables help quantify the trajectory space's expansion and ambiguity as a function of system parameters.

\begin{table*}[htbp]
\centering
\scriptsize
\caption{Summary of Simulation Results}
\label{tab:spip-results}
\resizebox{\textwidth}{!}{%
\begin{tabular}{llllllllll}
\toprule
\textbf{Experiment} & \textbf{Steps} & \textbf{Transforms} & \textbf{$\epsilon$} & \textbf{Entropy (bits)} & \textbf{Unique Endpoints} & \textbf{Collisions} & \textbf{Most Frequent Count} & \textbf{Avg Distance} & \textbf{Symbolic Freedom} \\
\midrule
Run 1 & 30   & 2   & 0.05 & 2.71 & 7  & 7  & 250 & 1.54 & 2.71 \\
Run 2 & 60   & 4   & 0.10 & 3.23 & 12 & 12 & 254 & 1.90 & 1.62 \\
Run 3 & 120  & 6   & 0.25 & 3.23 & 15 & 14 & 231 & 2.18 & 1.25 \\
Run 4 & 200  & 8   & 0.40 & 3.47 & 23 & 16 & 179 & 2.60 & 1.16 \\
Run 5 & 300  & 12  & 0.50 & 3.57 & 23 & 20 & 186 & 2.60 & 1.00 \\
Run 6 & 500  & 20  & 0.60 & 3.70 & 24 & 21 & 161 & 2.62 & 0.86 \\
Run 7 & 800  & 30  & 0.70 & 3.85 & 28 & 25 & 170 & 2.80 & 0.79 \\
Run 8 & 1200 & 40  & 0.80 & 3.94 & 29 & 26 & 159 & 2.87 & 0.74 \\
\bottomrule
\end{tabular}
}
\end{table*}

The table \ref{tab:spip-results} highlights several empirical phenomena that align closely with the theoretical model of SPIP:

First, we observe that symbolic entropy increases with trajectory depth, confirming the exponential growth of the trajectory space $\mathcal{W}_n$ outlined in Section~\ref{sec:spip}. However, this growth is not unbounded, entropy plateaus around 4 bits, suggesting the onset of symbolic saturation as overlapping trajectories collapse to shared endpoints.

Second, the number of unique endpoints rises with the number of steps, but in a sublinear fashion. This confirms the presence of symbolic collisions and non, injective mappings, which were formally established through complexity, theoretic reductions in Section~\ref{sec:complexity}.

Third, symbolic freedom, defined as normalized entropy per symbolic control choice, declines as the number of transformations increases. This counterintuitive trend underscores a core insight: adding symbolic options does not always yield more effective entropy, especially when contraction and noise induce trajectory convergence. These findings directly reinforce the SPIP-H assumption, which posits that symbolic inversion remains hard despite increasing symbolic flexibility.

To deepen this analysis, we now visualize how specific trajectory properties evolve across the simulated runs. These plots provide concrete, interpretable evidence of symbolic complexity growth and support the core cryptographic intuition behind SPIP.

\begin{figure}[htbp]
    \centering
     \includegraphics[width=\linewidth]{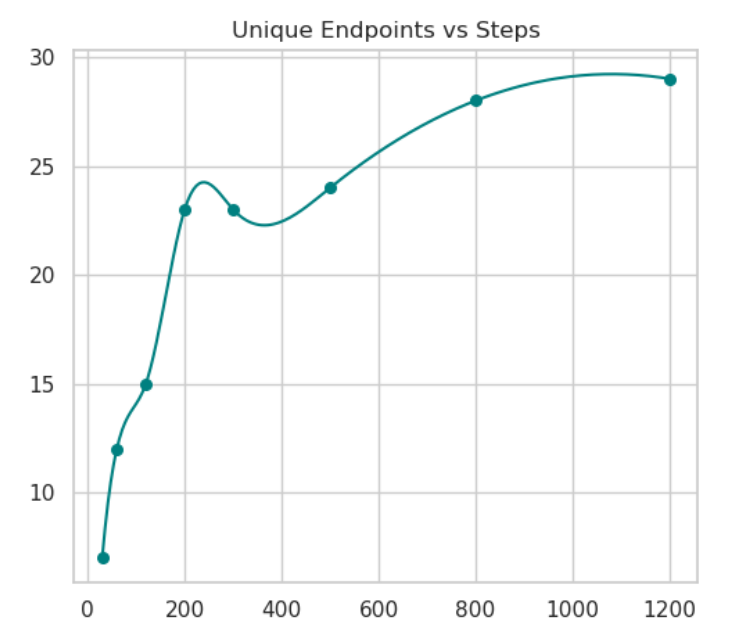}
    \caption{Unique endpoint growth vs. step count.}
    \label{fig:unique_endpoints}
\end{figure}

In Figure~\ref{fig:unique_endpoints}, we see that the number of unique endpoints initially increases rapidly with the number of steps, but begins to plateau after 800--1000 iterations. This saturation behavior, though empirically derived, aligns with our theoretical expectations from bounded rounding and contraction, which inevitably funnel trajectories into a finite symbolic basin.

\begin{figure}[htbp]
    \centering
     \includegraphics[width=\linewidth]{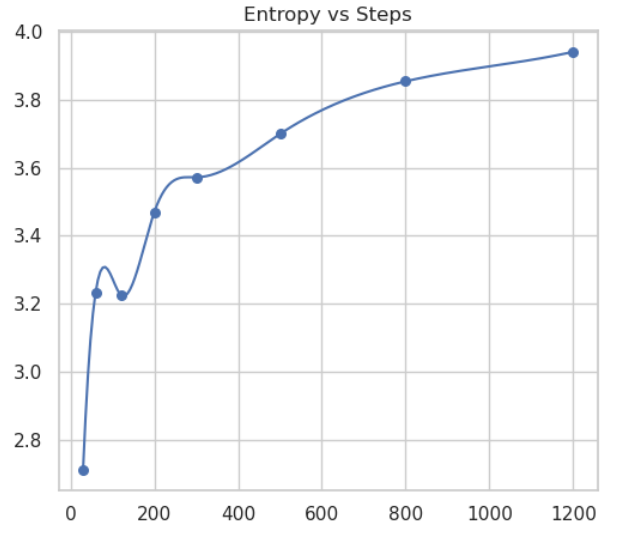}
    \caption{Entropy vs. average spatial dispersion.}
    \label{fig:entropy_vs_distance}
\end{figure}

Figure~\ref{fig:entropy_vs_distance} captures a key post-quantum consideration: as symbolic paths disperse more widely, the entropy of their endpoint distribution increases. The positive trend observed here suggests that symbolic chaos injects real cryptographic unpredictability, especially for long paths.

\begin{figure}[htbp]
    \centering
     \includegraphics[width=\linewidth]{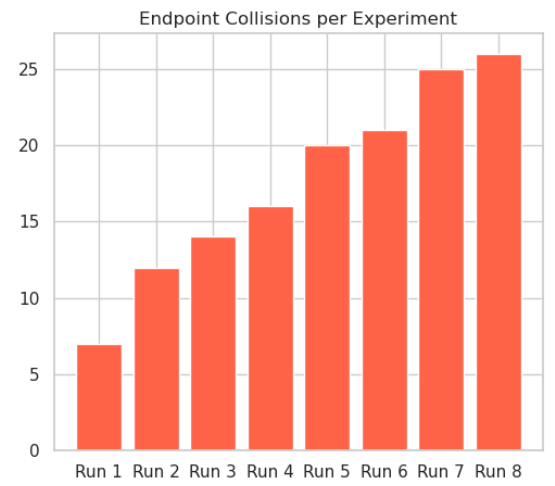}
    \caption{Collision count vs. experiment.}
    \label{fig:collisions}
\end{figure}

Figure~\ref{fig:collisions} shows a monotonic increase in endpoint collisions, affirming the many-to-one nature of SPIP mappings. This symbolic overlap is critical: it introduces ambiguity that prevents efficient inversion, not because of cryptographic hiding, but due to structural chaos.

\begin{figure}[htbp]
    \centering
     \includegraphics[width=\linewidth]{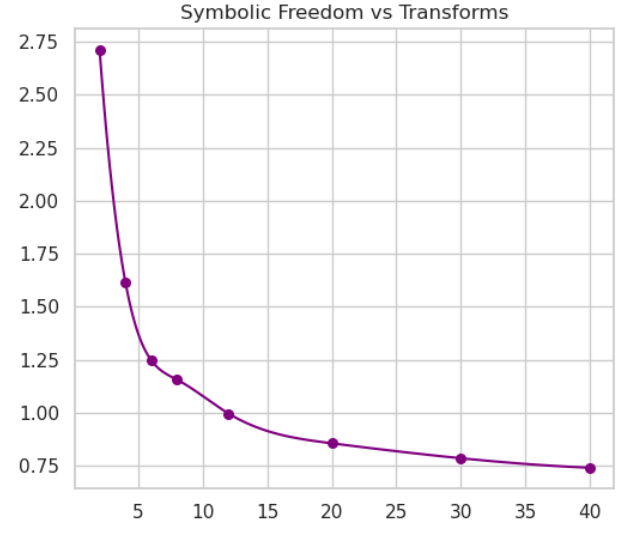}
    \caption{Symbolic freedom declines with transformation count.}
    \label{fig:symbolic_freedom}
\end{figure}

In Figure~\ref{fig:symbolic_freedom}, the decline in symbolic freedom reflects redundancy in symbolic control. As more transformations are introduced, the marginal entropy per transformation shrinks, a phenomenon that further deepens the inversion difficulty by embedding unpredictability without pattern regularity.

\begin{figure}[htbp]
    \centering
     \includegraphics[width=\linewidth]{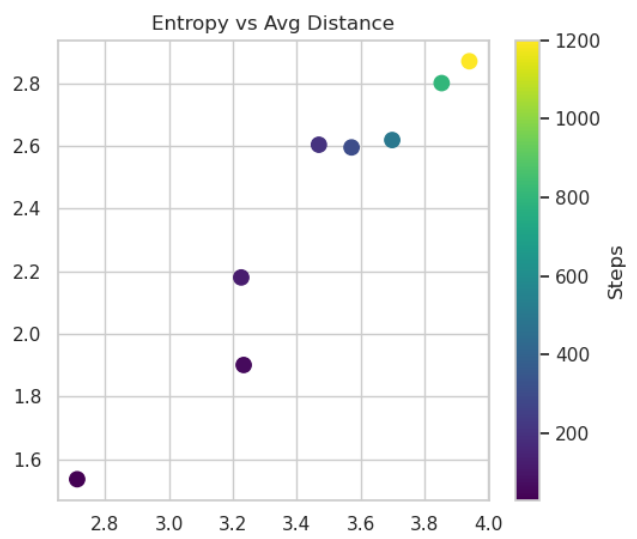}
    \caption{Entropy increases with symbolic depth.}
    \label{fig:entropy_vs_steps}
\end{figure}

Entropy growth is further confirmed in Figure~\ref{fig:entropy_vs_steps}, where symbolic complexity scales with path length, reinforcing SPIP’s alignment with the one-wayness intuition: the forward process expands unpredictably, while the backward recovery becomes infeasible due to entropy dispersion.

\begin{figure}[htbp]
    \centering
     \includegraphics[width=\linewidth]{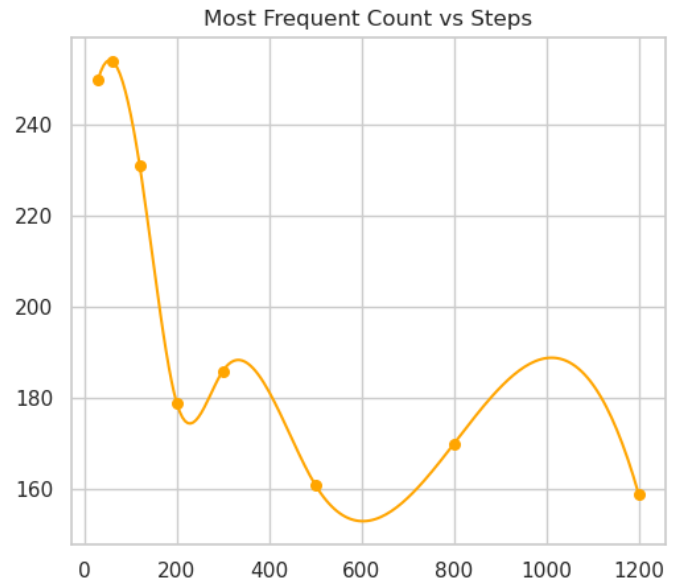}
    \caption{Dominance of most frequent endpoint vs. steps.}
    \label{fig:count_vs_steps}
\end{figure}

Finally, Figure~\ref{fig:count_vs_steps} illustrates the dominance of certain symbolic endpoints. Despite high entropy, a few attractor endpoints persistently capture many trajectories. This imbalance highlights the inherent symbolic bias in the system a further challenge to inversion, as it blurs distinguishability between high-probability and low-probability symbolic codes.

\vspace{0.5em}

Taken together, these results provide strong empirical support for the SPIP-H assumption. They validate the presence of symbolic explosion, nonlinear entropy growth, and endpoint ambiguity under simple, structure-free dynamics hallmarks of a one-way system in the information theoretic sense.

\section{Discussion and Implications} \label{sec:discussion}
\vspace{0.8em}

Having established the SPIP-H assumption and its complexity-theoretic grounding, we now examine its broader implications for cryptographic theory. We compare SPIP to existing hardness assumptions, explore its potential as a building block for cryptographic primitives, and outline possible limitations and open questions.

\subsection{Comparison with Classical Approach}

SPIP represents a conceptual departure from traditional hardness assumptions, which typically rely on algebraic structures such as finite fields or vector spaces as found in RSA, ECC, LWE, and multivariate cryptography \cite{shor1994algorithms, regev2009lattices, patarin1996hidden}. These structures, while enabling efficient constructions and formal reductions, also introduce symmetries that quantum algorithms like Shor’s can exploit.

In contrast, SPIP derives its hardness from symbolic chaos over $\mathbb{Z}^2$ via contractive affine maps with bounded noise. It lacks algebraic symmetry entirely there are no groups, rings, or fields to invert. The complexity stems from exponential path growth, non-injective mappings, and symbolic ambiguity, not from solving structured mathematical problems.

This absence of structure makes SPIP naturally resistant to known quantum attacks, including Shor’s and more subtly Grover’s, which faces verification ambiguity and oracle instability in SPIP’s combinatorial landscape. As such, SPIP opens the door to a distinct design space in post-quantum cryptography: one grounded in symbolic complexity rather than algebraic hardness.

\subsection{SPIP as Cryptographic Primitives?}

The SPIP assumption naturally lends itself to the construction of OWF. Given an initial point $x_0$, a sequence of contractive transformations $\mathcal{T}$, and path length $n$, the final point $x_n$ is computed via forward symbolic evolution:
\[
f(x_0) = x_n = T_{n-1} \circ \cdots \circ T_0 (x_0)
\]
Due to the symbolic ambiguity and exponential trajectory space, inverting this process i.e., recovering a valid preimage path from $x_n$, is assumed to be infeasible under the SPIP-H assumption.

Extending SPIP to PRF may require additional assumptions or hybrid constructions. One potential direction is to introduce secret keys through the initial seed $x_0$ or through a keyed transformation sequence $\mathcal{T}_k$. Such designs could define keyed mappings with high entropy, but without formal proofs of unpredictability or uniform output distribution, these constructions remain speculative.

\subsection{Limitations and Open Questions}

Despite its promise, SPIP raises several important challenges and open questions:

\begin{itemize}
    \item \textbf{Side-channel resilience:} If an adversary gains partial knowledge of internal states or noise vectors (e.g., via side-channel leakage), the inversion hardness could degrade significantly.
    
    \item \textbf{Entropy quantification:} A deeper statistical understanding of symbolic entropy, path overlap, and endpoint distributions is needed to support claims of randomness and unpredictability.
    
    \item \textbf{Symbol-to-bit encoding:} Translating symbolic trajectories into binary keys or ciphertexts requires encoding schemes that preserve one-wayness and avoid leakage or compression artifacts.
    
    \item \textbf{Quantum adversaries:} While SPIP lacks algebraic structure exploitable by current quantum algorithms, a formal quantum security model (e.g., under QROM) is still needed to establish robustness.
\end{itemize}

Nevertheless, SPIP introduces a new cryptographic design paradigm based on symbolic dynamics and chaotic evolution. Its structure-free nature offers inherent post-quantum resistance and invites exploration of entirely new primitives grounded in combinatorial complexity rather than algebra.

\section*{This is Not a Protocol Design Paper}
\vspace{0.8em}
To avoid misinterpretation, we emphasize that this paper does not present a cryptographic protocol, scheme, or practical instantiation. Instead, it introduces and analyzes a new computational hardness assumption the SPIP based on the combinatorial and dynamical complexity of symbolic chaotic walks over $\mathbb{Z}^2$.

we rigorously define SPIP, analyze its complexity-theoretic properties, and argue for its plausibility as a basis for structure-free, post-quantum cryptography. The paper does not include algorithmic instantiations, encoding mechanisms, implementation benchmarks, or security proofs for cryptographic constructions. These are left as future work.

We make this distinction explicit to ensure that the theoretical scope of this work is not confused with applied cryptographic engineering.

\section{Conclusion and Future Work} \label{sec:conclusion}
\vspace{0.8em}

This paper proposes a fundamentally new perspective for post-quantum cryptography, hardness grounded not in algebraic structures, but in the combinatorial and dynamical complexity of chaotic symbolic systems. Most existing hardness assumptions in cryptography, such as those based on integer factorization, discrete logarithms, and lattices rely on structured algebraic frameworks that are vulnerable to quantum algorithms exploiting their inherent symmetries. In contrast, our approach seeks to construct cryptographic hardness from the unpredictable evolution of symbolic trajectories generated by contractive affine maps with bounded noise and rounding, a system inherently resistant to algebraic reductions and quantum exploitation. \\\\ We introduced the SPIP as a concrete formulation of this idea. SPIP formalizes the challenge of recovering symbolic trajectories over $\mathbb{Z}^2$ from endpoints, and we proved that it is both PSPACE-hard and \#P-hard. Empirical simulations further support these theoretical results, demonstrating exponential growth in trajectory space and significant endpoint collisions, even for moderate parameters. These findings suggest that SPIP provides a plausible, structure-free basis for post-quantum cryptographic constructions.\\\\ This work builds directly on our prior study of RNIFS \cite{bouke2025rnifs}, which established the mathematical properties of symbolic chaos, such as entropy growth, fractal attractors, and symbolic ambiguity, in discrete dynamical systems. While the RNIFS study focused on the mathematical behavior of symbolic trajectories, the present work extends these insights into the cryptographic domain by defining SPIP as a hardness assumption. Together, these works establish a coherent research trajectory, from foundational mathematical analysis to the development of new cryptographic paradigms grounded in symbolic unpredictability.  \\\\Looking forward, our research agenda progresses in stages. The immediate priority is to instantiate SPIP into practical cryptographic primitives, including one-way functions, pseudorandom generators, and key derivation schemes. Building on this, a critical challenge is designing robust symbolic-to-binary encoding methods to translate chaotic paths into usable cryptographic keys. Formalizing the security of SPIP-based primitives under the Quantum Random Oracle Model (QROM) will be essential for establishing quantum resistance. Finally, investigating side-channel resilience, such as the impact of partial leakage or approximation errors, will be crucial for real-world deployment.

\vspace{0.8em}

\section*{Declarations}
\vspace{0.8em}
\begin{itemize}
  \item \textbf{Funding:} Not applicable.
  
  \item \textbf{Conflict of Interest:} The authors declare that there is no conflict of interest.

\item \textbf{Availability of Data and Materials: } 
This study did not use any external datasets. All experimental results were generated through custom simulations as described in the manuscript. The full simulation code, along with configuration files and analysis scripts, is publicly available at:  
\href{https://github.com/drbouke/SPIP}{\texttt{https://github.com/drbouke/SPIP}}.

  \item \textbf{Ethics Approval:} Not applicable.
\end{itemize}

\vspace{10em}

\section*{Author Biographies}
\begin{tcolorbox}[biobox]
  \begin{wrapfigure}{l}{0.35\linewidth}
    \vspace{-0.5em}
    \includegraphics[width=\linewidth]{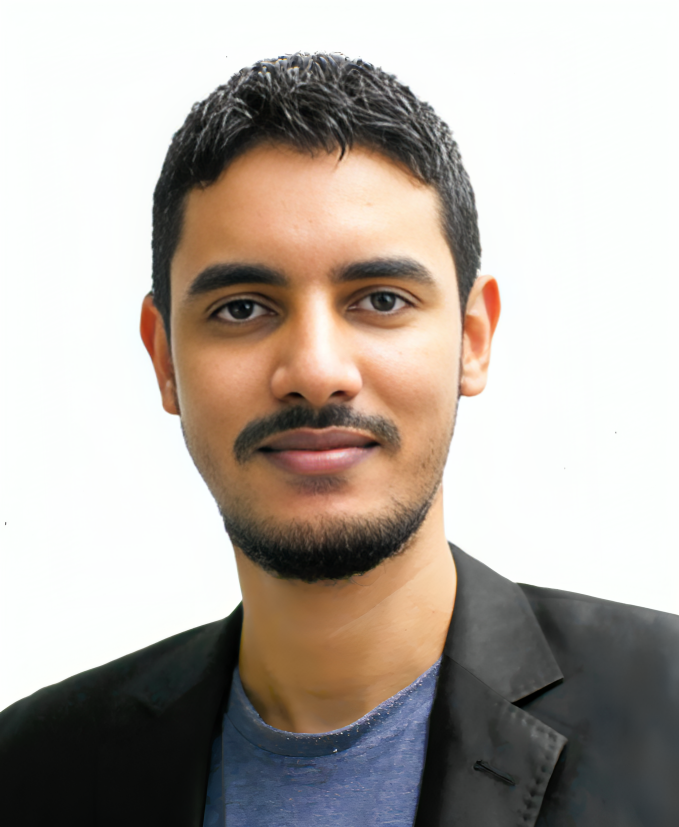}
  \end{wrapfigure}
  \textbf{Mohamed Aly Bouke}\,\href{https://orcid.org/0000-0003-3264-601X}{\includegraphics[height=1.8ex]{figs/orcid.png}} is a researcher with interdisciplinary expertise across theoretical mathematics, computer science, artificial intelligence, and cryptography. He holds a Master’s and a Ph.D. in Information Security from Universiti Putra Malaysia and has a background in mathematics education. His academic work spans topics such as mathematical modeling, epistemic systems, AI architectures, and secure computation. Dr. Bouke is an active member of the \textit{Institute of Electrical and Electronics Engineers (IEEE)}, the \textit{International Information System Security Certification Consortium (ISC2)}, and the \textit{Institute for Systems and Technologies of Information, Control and Communication (INSTICC)}. His contributions include peer-reviewed publications, invited talks, and academic training in both technical and theoretical domains.

  \vspace{0.8em}
  \textit{Email: bouke@ieee.org}
\end{tcolorbox}


\printbibliography

\end{document}